\documentclass{amsart}
\usepackage{amsmath}
  \usepackage{paralist}
  \usepackage{graphics} 
  \usepackage{graphicx}
  \usepackage{epsfig} 
 \usepackage[colorlinks=true]{hyperref}
\hypersetup{urlcolor=blue, citecolor=red}
\usepackage{array}
\usepackage{amsmath}
\usepackage{amssymb}
\usepackage{upref}
\usepackage{amsfonts}
\usepackage{graphicx}
\usepackage{wasysym}

\usepackage{mathrsfs}

\newcommand{\cal}[1]{\mathcal #1}

\newcommand{\deff}{\mbox{$\stackrel{\rm def}{=}$}}

\newcommand{\sbinom}[2]{\left[ \begin{array}{c} #1 \\ #2 \end{array} \right] }
\newcommand{\sbinomq}[2]{\sbinom{#1}{#2}_q }

\newcommand{\field}[1]{\mathbb{#1}}
\newcommand{\C}{\field{C}}
\newcommand{\F}{\field{F}}

\newcommand{\cF}{{\cal F}}

\newcommand{\cC}{{\cal C}}
\newcommand{\cG}{{\cal G}}

\newcommand{\cP}{{\cal P}}

\newcommand{\sP}{\cP}
\newcommand{\sG}{\cG}

\newcommand{\Gr}{\smash{{\sG\kern-1.5pt}_q\kern-0.5pt(n,k)}}
\newcommand{\Gk}{\smash{{\sG\kern-1.5pt}_q\kern-0.5pt(n,k_1)}}
\newcommand{\Gkk}{\smash{{\sG\kern-1.5pt}_q\kern-0.5pt(n,k_2)}}
\newcommand{\Grtwo}{\smash{{\sG\kern-1.5pt}_2\kern-0.5pt(n,k)}}
\newcommand{\Gkone}{\smash{{\sG\kern-1.5pt}_q\kern-0.5pt(n,k_1)}}
\newcommand{\Gktwo}{\smash{{\sG\kern-1.5pt}_q\kern-0.5pt(n,k_2)}}
\newcommand{\Ps}{\smash{{\sP\kern-2.0pt}_q\kern-0.5pt(n)}}

\newtheorem{theorem}{Theorem}[section]
\newtheorem{corollary}{Corollary}

\newtheorem{lemma}[theorem]{Lemma}

\theoremstyle{definition}

\newtheorem{remark}{Remark}
\newtheorem{example}{Example}

\DeclareMathOperator{\rank}{rank}

\title[Large Constant Dimension Codes and Lexicodes]
      {Large Constant Dimension Codes and Lexicodes}

\author[Natalia Silberstein and Tuvi Etzion]{}


 \email{natalys@cs.technion.ac.il}
 \email{etzion@cs.technion.ac.il}

\thanks{This work was supported in part by the Israel
Science Foundation (ISF), Jerusalem, Israel, under Grant 230/08.}

\begin{document}
\maketitle

\centerline{\scshape Natalia Silberstein }
\medskip
{\footnotesize
 \centerline{Computer Science Department}
   \centerline{Technion - Israel Institute of Technology}
   \centerline{ Haifa, Israel, 32000}
} 

\medskip

\centerline{\scshape   Tuvi Etzion}
\medskip
{\footnotesize
 \centerline{ Computer Science Department}
   \centerline{Technion - Israel Institute of Technology}
   \centerline{ Haifa, Israel, 32000}
}

\bigskip


\begin{abstract}
Constant dimension codes, with a prescribed minimum distance, have
found recently an application in network coding. All the codewords
in such a code are subspaces of $\F_q^n$ with a given dimension. A
computer search for large constant dimension codes is usually
inefficient since the search space domain is extremely large. Even
so, we found that some constant dimension lexicodes are larger
than other known codes. We show how to make the computer search
more efficient. In this context we present a formula for the
computation of the distance between two subspaces, not necessarily
of the same dimension.
\end{abstract}

\section{Introduction}

\label{sec:introduction}
Let $\F_q$ be the finite field of size $q$. The set of all $k$-dimensional
subspaces of  the vector space~\smash{$\F_q^n$}, for any given two
nonnegative integers $k$ and $n$, $0\leq k \le n$, forms the
\emph{Grassmannian space} (Grassmannian, in short) over  $\F_q$,
denoted by $\Gr$. The Grassmannian space is a metric space, where the
\emph{subspace distance} between any two subspaces $X$ and $Y$ in $\Gr$,
is given by
\begin{equation}
\label{def_subspace_distance}
d_S (X,\!Y) \,\ \deff\ \dim X + \dim Y -2 \dim\bigl( X\, {\cap}Y\bigr).
\end{equation}
This is also the definition for the distance between two subspaces of
~\smash{$\F_q^n$}
which are not of the same dimension.

We say that $\C\subseteq \Gr$ is an $(n,M,d,k)_q$ \emph{code in
the Grassmannian}, or \emph{constant-dimension code}, if $M =
|\C|$ and $d_S (X,\!Y) \ge d$ for all distinct elements $X,\!Y \in
\C$. The minimum distance of $\C$, $d_S(\C)$, is $d$.

Koetter and Kschischang~\cite{KK} presented an application of
error-correcting codes in  $\Gr$ to random network coding. This
led to an extensive research for construction of large codes in
the Grassmannian. Constructions and bounds for such codes were
given in~\cite{EtSi09,EV,GaYa08,KK,KoKu08,SiEt09,SKK08,Ska08}.

The motivation for this work is an $(8,4605,4,4)_2$ constant
dimension lexicode constructed in~\cite{SiEt09} which is larger
than any other known codes with the same parameters.

\textit{Lexicographic codes}, or \textit{lexicodes}, are greedily
generated error-correcting codes which were first developed by
Levinshtein~\cite{Lev60}, and rediscovered by Conway and
Sloane~\cite{CoSl86}. The construction  of a lexicode with a
minimum distance $d$ starts with the set $ \mathcal {S}= \{S_0\}$,
where $S_0$ is the first element in a lexicographic order, and
greedily adds the lexicographically first element whose distance
from all the elements of $\mathcal {S}$ is at least $d$. In the
Hamming space, the lexicodes include the optimal codes, such as
the Hamming codes and the Golay codes.

To construct a lexicode, we need first to define some order of all
subspaces in the Grassmannian. The $(8,4605,4,4)_2$ lexicode found
in~\cite{SiEt09} is based on the Ferrers tableaux form
representation of a subspace. First, for completeness, we provide
the definitions which are required to define the Ferrers tableaux
form representation of subspaces in the Grassmannian, and next, we
define the order of the Grassmannian based on this representation.

A \textit{partition} of a positive integer $m$ is a representation
of $m$ as  a sum of positive integers, not necessarily distinct.

A {\it Ferrers diagram} $\cF$ represents a partition as a pattern
of dots with the $i$-th row having the same number of dots as the
$i$-th term in the partition~\cite{AnEr04,vLWi92,Sta86}. (In the
sequel, a {\it dot} will be denoted by a $"\bullet"$). A Ferrers
diagram satisfies the following conditions.
\begin{itemize}
\item The number of dots in a row is at most the number of dots in
the previous row.

\item All the dots are shifted to the right of the diagram.
\end{itemize}

The {\it number of rows (columns)} of the Ferrers diagram $\cF$ is
the number of dots in the rightmost column (top row) of $\cF$. If
the number of rows in the Ferrers diagram is $m$ and the number of
columns is $\eta$, we say that it is an $m \times \eta$ Ferrers
diagram.

Let $X\in \Gr$ be a $k$-dimensional subspace in the Grassmannian.
We can represent $X$ by the $k$ linearly independent vectors from
$X$ which form  a unique $k\times n$ generator matrix in {\it
reduced row echelon form} (RREF), denoted by $RE(X)$, and defined
as follows:
\begin{itemize}
\item The leading coefficient of a row is always to the right of
the leading coefficient of the previous row.

\item All leading coefficients are {\it ones}.

\item Every leading coefficient is the only nonzero entry in its
column.
\end{itemize}

For each $X\in \Gr$  we associate a binary vector of length $n$
and weight $k$, $v(X)$, called the \emph{identifying vector} of
$X$, where the \textit{ones} in $v(X)$ are exactly in the
positions where $RE(X)$ has the leading \textit{ones}.

The {\it echelon Ferrers form} of a binary vector $v$ of length
$n$ and weight $k$, $EF(v)$, is the $k\times n$ matrix in RREF
with leading entries (of rows) in the columns indexed by the
nonzero entries of $v$ and $"\bullet"$  in all entries which do
not have terminal {\it zeroes} or {\it ones} (see~\cite{EtSi09}).
The dots of this matrix form the Ferrers diagram $\cF$ of $EF(v)$.
If we substitute elements of $\F_q$ in the dots of $EF(v)$ we
obtain a generator matrix in RREF of a $k$-dimensional subspace of
$\Gr$. $EF(v)$ and $\cF$ will be called also the echelon Ferrers
form  and the Ferrers diagram of such a subspace, respectively.

The {\it Ferrers tableaux form} of a subspace $X$, denoted by
$\cF(X)$, is obtained by assigning the values of $RE(X)$ in the
Ferrers diagram $\cF_X$ of $X$. Each Ferrers tableaux form
represents a unique subspace in $\Gr$.

\begin{example}  Let $X$ be the subspace in $\mathcal G_2(7,3)$ with
the  following  generator matrix in RREF:
$$RE(X)=\left( \begin{array}{ccccccc}
1 & 0 & 0 & 0 & 1 & 1 & 0 \\
0 & 0 & 1 & 0 & 1 & 0 & 1 \\
0 & 0 & 0 & 1 & 0 & 1 & 1
\end{array}
\right) ~.$$ Its identifying vector is $v(X)=1011000$, and its
echelon Ferrers form, Ferrers diagram,   and Ferrers tableaux form
are given by
$$\left[ \begin{array}{ccccccc}
1 & \bullet & 0 & 0 & \bullet & \bullet & \bullet \\
0 & 0 & 1 & 0 & \bullet & \bullet & \bullet \\
0 & 0 & 0 & 1 & \bullet & \bullet & \bullet
\end{array}\right],~~~
\begin{array}{cccc}
 \bullet & \bullet & \bullet & \bullet \\
  & \bullet & \bullet & \bullet   \\
  & \bullet & \bullet & \bullet  \\
\end{array},~~~
\; \textrm{and }\;
\begin{array}{cccc}
0 & 1 & 1 & 0 \\
&1 & 0 & 1  \\
&0 & 1 & 1
\end{array},\; \textrm{respectively }.$$
\end{example}

Let $\cF$ be a Ferrers diagram  of a subspace $X\in \Gr$.  $\cF$ can be
embedded in  a
$k\times (n-k)$ box. We represent $\cF$ by an integer vector of
length $n-k$, $(\cF_{n-k},...,\cF_2,\cF_1)$, where $\cF_i$ is
equal to the number of dots in the $i$-th column of $\cF$, $1\leq
i\leq n-k$, where we number the columns from right to left. Note
that $\cF_{i+1} \leq \cF_i$, $1 \leq i \leq n-k-1$.

To define an order of all the subspaces in the Grassmannian we
need first to define an order of all the Ferrers diagrams embedded
in the $k\times (n-k)$ box.
Let $| \cF |$ denote the {\it size} of $\cF$, i.e., the number of
dots in $\cF$. For two Ferrers diagrams $\cF$ and
$\widetilde{\cF}$, we say that $\cF < \widetilde{\cF}$ if one of
the following two conditions holds.
\begin{itemize}
\item $|\cF| > |\widetilde{\cF}|$

\item $|\cF| = |\widetilde{\cF}|$,  and  $\cF_i > \widetilde{\cF}_i$
for the least index $i$ where the two diagrams $\cF$  and
$\widetilde{\cF}$ have a different number of dots.
\end{itemize}

Now, we define the following order of subspaces in the
Grassmannian based on the Ferrers tableaux form representation.
Let $X$, $Y\in\Gr$ be two $k$-dimensional subspaces,  and $RE(X)$,
$RE(Y)$ their related RREFs. Let $v(X),~v(Y)$ be the identifying
vectors of $X,~Y$, respectively, and $\cF_X,~\cF_Y$ the Ferrers
diagrams of $EF(v(X)),~EF(v(Y))$, respectively. Let
$x_1,x_2,...,x_{|\cF_X|}$ and $y_1,y_2,...,y_{|\cF_Y|}$ be the
entries of Ferrers tableaux forms $\cF(X)$ and $\cF(Y),$
respectively. The entries of a Ferrers tableaux form are numbered
from right to left, and from top to bottom.

We say that $X < Y$ if one of the following two conditions holds.
\begin{itemize}
\item $\cF_X < \cF_Y; $

\item $\cF_X = \cF_Y $,  and $(x_1,x_2,...,x_{|\cF_X|}) <
(y_1,y_2,...,y_{|\cF_Y|}).$
\end{itemize}
\vspace{0.2cm}
\begin{example}
Let $X,Y,Z,W\in\mathcal G_2(6,3)$ be given by
\begin{footnotesize}
\begin{align*}
\cF(X)=
\begin{array}{ccc}
1 & 1 & 1 \\
1 & 1 & 1  \\
&  & 1
\end{array},~~~
\cF(Y)=
\begin{array}{ccc}
1 & 0 & 1 \\
& 0 & 0  \\
& 1 & 1
\end{array},~~~
\cF(Z)=
\begin{array}{ccc}
1 & 1 & 1 \\
& 1 & 1  \\
& & 0
\end{array},~~~
\cF(W)=
\begin{array}{ccc}
1 & 1 & 1 \\
& 1 & 1  \\
& & 1
\end{array}.
\end{align*}
\end{footnotesize}
By the definition, we have that $\cF_Y < \cF_X < \cF_Z =\cF_W $.
Since $(z_1,z_2,...,z_{|\cF_Z|})=(1,1,0,1,1,1)<
(w_1,...,w_{|\cF_W|})=(1,1,1,1,1,1)$, it follows that $Y<X<Z<W$.
\end{example}

The construction of lexicodes involves many computations of the
distance between two subspaces of $\Gr$. In
Section~\ref{sec:distance} we develop a new formula for
computation of the distance between two subspaces not necessarily
of the same dimension. This formula will enable a faster
computation of the distance between any two subspaces of $\Gr$. In
Section~\ref{sec:analysis} we examine several properties of
constant dimension codes which will enable to simplify the
computer search for large lexicodes. In Section~\ref{sec:search}
we describe a general search method for constant dimension
lexicodes. We also present some improvements on the sizes of
constant dimension codes. In Section~\ref{sec:conclusion}  we
summarize our results and present several problems for further
research.

\section{Computation of Distance between Subspaces}
\label{sec:distance}

The research on error-correcting codes in the Grassmannian in
general and on the search for related lexicodes in particular
requires many computations of the distance between two subspaces
in the Grassmannian. We will examine  a more general computation
problem of the distance between any two subspaces $X,Y\subseteq
\F_q^n$  which do not necessarily have the same dimension. The
motivation is to simplify the computations that lead to the next
subspace which will be joined to the lexicode.

Let $A*B$ denotes the concatenation  $\left(\begin{footnotesize}
\begin{array}{c}
A\\B\end{array}\end{footnotesize}\right)$ of two matrices $A$ and $B$
with the same number of columns.
By the definition of the subspace distance (\ref{def_subspace_distance}),
it follows that
\begin{align}
d_S(X,Y) & =  2\rank(RE(X)*RE(Y))
 -  \rank(RE(X)) - \rank(RE(Y)).\label{distance-rank}
\end{align}

Therefore, the calculation of $d_S(X,Y)$ can be done by using
Gauss elimination. In this section we present an improvement on
this calculation by using the  representation of subspaces by
Ferrers tableaux forms, from which their identifying vectors  and
their RREF are easily determined. We will present an alternative
formula for the computation of the distance between two subspaces
$X$ and $Y$.

For $X \in \Gk$ and $Y \in \Gkk$, let  $\rho(X,Y)$ [$\mu(X,Y)$] be
a set of indices (of coordinates) with common \textit{zeroes}
[\textit{ones}] in $v(X)$  and $v(Y)$, i.e.,
\[
 \rho(X,Y)=\left\{  i |\;v(X)_i=0\mbox{ and }v(Y)_i
=0\right\} ,\]
 and \[
\mu(X,Y) =\left\{  i|\;v(X)_i=1\mbox{ and }v(Y)_i
=1\right\} .\]

Note that $|\rho(X,Y)|+|\mu(X,Y)|+d_{H}(v(X),v(Y))=n$, where $d_H(\cdot ,\cdot)$
denotes the Hamming distance, and
\begin{equation}
|\mu(X,Y)|=\frac{k_1+k_2-d_{H}(v(X),v(Y))}{2}.\label{commom_ones}
\end{equation}

Let $X_\mu$ be the $|\mu(X,Y)| \times n$ sub-matrix of $RE(X)$
which consists of the rows with leading \textit{ones} in the
columns related to (indexed by) $\mu(X,Y)$. Let $X_{\mu^C}$  be
the $(k_1-|\mu(X,Y)|)\times n$ sub-matrix of $RE(X)$  which
consists of all the rows  of  $RE(X)$ which are not contained in
$X_\mu$. Similarly, let $Y_{\mu}$ be the $|\mu(X,Y)| \times n$
sub-matrix of $RE(Y)$ which consists of the rows with leading
\textit{ones} in the columns related to  $\mu(X,Y)$. Let
$Y_{\mu^C}$ be the $(k_2-|\mu(X,Y)|)\times n$ sub-matrix of
$RE(Y)$ which consists of all the rows  of $RE(Y)$ which are not
contained in $Y_{\mu}$.

Let $\widetilde{X}_\mu$  be the $|\mu(X,Y)|\times n$ sub-matrix of
$RE(RE(X) * Y_{\mu^C})$ which consists of the rows with
leading\textit{ ones} in the columns indexed by  $\mu(X,Y)$.
Intuitively, $\widetilde{X}_\mu$ obtained by concatenation of the
two matrices, $RE(X)$ and $Y_{\mu^C}$,  and "cleaning" (by adding
the corresponding rows of $Y_{\mu^C}$) all the nonzero entries in
columns of $RE(X)$ indexed by leading \textit{ones} in
$Y_{\mu^C}$. Finally, $\widetilde{X}_\mu$ is obtained by taking
only the rows which are indexed by $\mu(X,Y)$. Thus,
$\widetilde{X}_\mu$ has all-zeroes columns indexed by
\textit{ones} of $v(Y)$  and  $v(X)$ which are not in $\mu(X,Y)$.
Hence $\widetilde{X}_\mu$ has nonzero elements only in columns
indexed by $\rho(X,Y)\cup\mu(X,Y)$.

Let $\widetilde{Y}_\mu$ be the $|\mu(X,Y)|\times n$ sub-matrix of
$RE(RE(Y) * X_{\mu^C})$ which consists of the rows with
leading\textit{ ones} in the columns indexed by  $\mu(X,Y)$.
Similarly to $\widetilde{X}_\mu$, it can be verified that
$\widetilde{Y}_\mu$ has nonzero elements only in columns indexed
by $\rho(X,Y)\cup\mu(X,Y)$.

\begin{corollary}
\label{cor:size of X-Y}
Nonzero entries in  $\widetilde{X}_\mu - \widetilde{Y}_\mu$
can appear only in columns indexed by $\rho(X,Y)$.
\end{corollary}
\begin{proof}
An immediate consequence since the columns of $\widetilde{X}_\mu$
and $\widetilde{Y}_\mu$ indexed by $\mu(X,Y)$ form a
$|\mu(X,Y)|\times |\mu(X,Y)|$ identity matrix.
\end{proof}

\begin{theorem}
\label{thm:distance}
\begin{equation}
d_{S}(X,Y)=d_H (v(X),v(Y))+2\rank (\widetilde{X}_\mu - \widetilde{Y}_\mu)
\label{subspace_distance}.
\end{equation}
\end{theorem}

\begin{proof}
By (\ref{distance-rank}) it is sufficient to proof that

\begin{align}
2\rank (RE(X)*RE(Y))& = k_1+k_2+d_H(v(X),v(Y))
+2\rank(\widetilde{X}_\mu - \widetilde{Y}_\mu).\label{eq.thm}
\end{align}

It is easy to verify that

\begin{align}
\rank
\left(\begin{array}{c}
RE(X)\\RE(Y)\end{array}\right)= \rank \left(\begin{array}{c}
RE(X)\\Y_{\mu^C}\\Y_\mu\end{array}\right)
=\rank \left(\begin{array}{c}
RE(X)\\Y_{\mu^C}\\\widetilde{Y}_\mu\end{array}\right)\nonumber\\
=\rank\left(\begin{array}{c}
RE(RE(X)*Y_{\mu^C})\\
\widetilde{Y}_\mu\end{array}\right)
=\rank \left(\begin{array}{c}
RE(RE(X)*Y_{\mu^C})\\
\widetilde{Y}_\mu-\widetilde{X}_\mu\end{array}\right)\label{eq:rank}.
\end{align}

We note that the positions of the leading \textit{ones} in all the
rows of $RE(X)*Y_{\mu^C}$ are in $\{1,2,\ldots,n\}\setminus
\rho(X,Y)$. By Corollary \ref{cor:size of X-Y} the positions of
the leading \textit{ones } of all the rows of
$RE(\widetilde{Y}_\mu-\widetilde{X}_\mu)$ are in $\rho(X,Y)$.
Thus, by (\ref{eq:rank}) we have
\begin{align}
\rank (RE(X)*RE(Y)) =&\rank(RE(RE(X)*Y_{\mu^C})+\rank
(\widetilde{Y}_\mu-\widetilde{X}_\mu).\label{eq:1}
\end{align}

Since the sets of positions of the leading \textit{ones} of
$RE(X)$ and $Y_{\mu^C}$ are disjoint, we have that $\rank
(RE(X)*Y_{\mu^C})= k_1+(k_2-|\mu(X,Y)|)$, and thus by (\ref{eq:1})

\begin{align}
\rank (RE(X)*RE(Y))
=&k_1+k_2-|\mu(X,Y)|+\rank(\widetilde{Y}_\mu-\widetilde{X}_\mu).
\label{eq:2}
\end{align}
Combining (\ref{eq:2}) and (\ref{commom_ones}) we have
\begin{align*}
2\rank (RE(X)*RE(Y))
= k_{1}+k_{2}+d_{H}(v(X),v(Y))+2\rank (\widetilde{Y}_\mu-\widetilde{X}_\mu),
\end{align*}
and by (\ref{eq.thm}) this proves the theorem.

\end{proof}


\begin{corollary}
\label{cor:distance}
For any two subspaces $X,Y\subseteq \F_q^n$,
$$d_{S}(X,Y)\geq d_{H}(v(X),v(Y)).$$
\end{corollary}

\begin{corollary}
\label{cor:distanceSameID}
Let $X$  and $Y$ be two subspaces such that $v(X)=v(Y)$. Then
$$d_{S}(X,Y)= 2\rank(RE(X)-RE(Y)).$$
\end{corollary}

In the sequel we will show how we can use these results to make the
search of lexicodes more efficient.


\section{Analysis of Constant Dimension Codes}
\label{sec:analysis}

In this section we consider some properties of constant dimension
codes which will help us to simplify the search for lexicodes.
First, we introduce the multilevel structure of a code in the
Grassmannian.

All the  binary vectors of the length $n$ and weight $k$
can be considered as  the identifying vectors of all the subspaces
in  $\Gr$. These  $\binom{n}{k}$ vectors partition  $\Gr$
into the $\binom{n}{k}$ different classes, where each class
consists of all subspaces in  $\Gr$ with the same identifying
vector. These classes are called \emph{Schubert cells}~\cite[p. 147]{Ful97}.
Note that each Schubert cell contains all the subspaces with the same
given echelon Ferrers form.

According to this  partition all the constant dimension
codes have a multilevel structure: we can partition  all
the codewords  of a code into different classes (sub-codes), each of which
have the same identifying vector. Therefore, the first level of this
structure is the set of different identifying vectors, and the second
level is the subspaces corresponding to these vectors.

Let $\C\subseteq \Gr$ be a constant dimension code, and let
$\{v_1,v_2,\ldots,v_{t}\}$ be all the different identifying
vectors of the codewords in $\C$. Let $\{\C_1,\C_2,\ldots,\C_t\}$
be the partition of $\C$ into $t$ sub-codes induced by these $t$
identifying vectors, i.e., $v(X)=v_i$, for each $X\in \C_i$,
$1\leq i\leq t$.

\begin{remark} We can choose any constant weight code $ \textit{C}$  with  minimum
Hamming distance $d$ to be the set of identifying vectors. If for
each identifying vector $v\in \textit{C}$  we have a sub-code
$\C_v$ for which $v(X)=v$  for each $X\in \C_v$, and
$d_S(\C_v)=d$, then by Corollary \ref{cor:distance} we obtain a
constant dimension code with the same minimum distance $d$. If for
all such identifying vectors we construct the maximum size
constant dimension sub-codes then we obtain the multilevel
construction (ML construction, in short) which was described in
~\cite{EtSi09}. One question that arises in this context is how to
choose the best constant weight code for this ML construction.
\end{remark}

To understand the structure of a sub-code formed by some Ferrers
diagram induced by an identifying vector, we need the following
definitions.

For two $m \times \eta$ matrices $A$ and $B$ over $\F_q$ the {\it
rank distance}, $d_R(A,B)$, is defined by
$$
d_R (A,B) ~ \deff ~ \text{rank}(A-B)~.
$$
A code $\cC$ is an $[m \times \eta,\varrho,\delta]$ {\it
rank-metric code} if its codewords are $m \times \eta$ matrices
over $\F_q$, they form a linear subspace of dimension $\varrho$ of
$\F_q^{m \times \eta}$, and for each two distinct codewords $A$
and $B$ we have that $d_R (A,B) \geq \delta$. For an $[m \times
\eta,\varrho,\delta]$ rank-metric code $\cC$ we have $\varrho \leq
\text{min}\{m(\eta-\delta+1),\eta(m-\delta+1)\}$
(see~\cite{Del78,Gab85,Rot91}). This bound is attained for all
possible parameters and the codes which attain it are called {\it
maximum rank distance} codes (or MRD codes in short).

Let $v$ be a vector of length $n$ and weight $k$ and let $EF(v)$
be its echelon Ferrers form. Let $\cF$ be the Ferrers diagram of
$EF(v)$. $\cF$ is an $m \times \eta$ Ferrers diagram, $m \leq k$,
$\eta \leq n-k$. A code $\cC$ is an $[\cF,\varrho,\delta]$ {\it
Ferrers diagram rank-metric code} if all codewords of $\cC$ are $m
\times \eta$ matrices in which all entries not in $\cF$ are {\it
zeroes}, it forms a rank-metric code with dimension $\varrho$, and
minimum rank distance $\delta$. Let $\dim (\cF,\delta)$ be the
largest possible dimension of an $[\cF,\varrho,\delta]$ code. The
following theorem~\cite{EtSi09} provides an upper bound on the
size of such codes.

\begin{theorem}
\label{thm:upper_rank} For a given $i$, $0 \leq i \leq \delta -1$,
if $\nu_i$ is the number of dots in $\cF$, which are not contained
in the first $i$ rows and are not contained in the rightmost
$\delta-1-i$ columns, then $\min_i \{ \nu_i \}$ is an upper bound
on $\dim (\cF,\delta)$.
\end{theorem}

It is not known  whether  the upper bound of
Theorem~\ref{thm:upper_rank} is attained for all parameters. A
code which attains this bound, will be called an MRD (Ferrers
diagram) code. This definition generalizes the previous definition
of MRD codes, and a construction of such codes is given
in~\cite{EtSi09}.

Without loss of generality we will
assume that $k \leq n-k$. This assumption
can be justified  as a consequence of the following
lemma~\cite{EV}.

\begin{lemma}
If $\C$ is an $(n,M,d ,k)_q$ constant dimension code then
$\C^\perp = \{ X^\perp : X \in\C\}$, where $X^\perp$ is the
orthogonal subspace of $X$, is an $(n,M,d ,n-k)_q$ constant
dimension code.
\end{lemma}

For $X\in \Gr$, we define the $k\times (n-k)$ matrix $R(X)$ as the
sub-matrix of $RE(X)$ with the columns which are indexed
by\textit{ zeroes} of $v(X)$. By
Corollary~\ref{cor:distanceSameID}, for any two codewords $X,Y\in
\C_i$, $\C_i \subseteq \C$, $1\leq i\leq t$, the subspace distance
between $X$ and $Y$ can be calculated in terms of rank distance,
i.e.,
$$d_S(X,Y)=2d_R(R(X),R(Y)).$$

For each sub-code $\C_i \subseteq \C$, $1\leq i\leq t$, we define
a Ferrers diagram rank-metric code
$$ R(\C_i)\,\ \deff\ \,\{R(X):X\in \C_i\}.$$
Note, that such a code is obtained by the inverse operation to the
\textit{lifting} operation, defined in ~\cite{SKK08}. Thus,
$R(\C_i)$ will be called the \textit{unlifted code} of the
sub-code $\C_i$.

We define the subspace distance between two sub-codes $\C_i$, $\C_j$
of $\C$, $1\leq i\neq j \leq t$ as follows:
$$d_S(\C_i,\C_j) =\min\{d_S(X,Y):X\in \C_i, Y\in \C_j\}.$$
By Corollary~\ref{cor:distance},
$$d_S(\C_i,\C_j)\geq d_H(v_i,v_j).$$

The following lemma shows a case in which the last inequality becomes an
equality.
\begin{lemma}
Let $\C_i$  and $\C_j$ be two different sub-codes of $\C\subseteq
\Gr$, each one contains the subspace whose RREF is the
corresponding column permutation of the matrix $(I_k 0_{k\times
(n-k)})$,  where $I_k$ denotes the $k\times k$ identity matrix and
$0_{a\times b}$  denotes an $a\times b$ allzero matrix. Then
$$ d_S(\C_i,\C_j)=d_H(v_i,v_j).$$
\end{lemma}

\begin{proof} Let $X\in \C_i$ and $Y\in \C_j$ be subspaces whose
RREF equal to column permutation of the matrix
$(I_k 0_{k\times (n-k)})$. It is easy to verify that
\begin{align}
\rank
\left(\begin{array}{c}
RE(X)\\RE(Y)\end{array}\right)= \rank \left(\begin{array}{c}
RE(X)\\Y_{\mu^C}\\Y_\mu\end{array}\right)=\rank \left(\begin{array}{c}
RE(X)\\Y_{\mu^C}\end{array}\right).
\end{align}
Clearly,  $\rank(Y_{\mu^C})=\frac{d_H(v_i,v_j)}{2}$, and hence,
$\text{rank}(RE(X)*RE(Y))=k+\frac{d_H(v_i,v_j)}{2}$.
By (\ref{distance-rank}), $d_S(X,Y)=2\text{rank}(RE(X)*RE(Y))-2k=2k+d_H(v_i,v_j)-2k=
d_H(v_i,v_j)$, i.e., $d_S(\C_i,\C_j)\leq d_H(v_i,v_j)$.  By Corollary
\ref{cor:distance}, $ d_S(\C_i,\C_j)\geq d_H(v_i,v_j)$, and hence, $d_S(\C_i,\C_j)=d_H(v_i,v_j)$.
\end{proof}

\begin{corollary}
\label{cor:linearity} Let $v_i$ and $v_j$ be two identifying
vectors of codewords in an $(n,M,d,k)_q$ code $\C$. If $d_H(v_i,
v_j) < d$ then at least one of the corresponding sub-codes, $\C_i$
and $\C_j$,  does not contain the subspace with RREF which is a
column permutation of the matrix $(I_k 0_{k\times(n-k)})$. In
other words, the corresponding unlifted code is not linear since
it does not contain the allzero  codeword.
\end{corollary}

Assume that we can add codewords to a code $\C$, $d_S(\C)=d$,
constructed by the ML construction with a maximal constant weight
code (for the identifying vectors) $\textit{C}$,
$d_H(\textit{C})=d$. Corollary~\ref{cor:linearity} implies that
any corresponding unlifted Ferrers diagram rank-metric code of any
new identifying vector will be nonlinear.

The next two lemmas reduce the search  domain for constant dimension lexicodes.

\begin{lemma}
\label{lm:first_id} Let $\C$ be an $(n,M,d=2\delta,k)_q$ constant
dimension code. Let $\C_1\subseteq\C$, $v(X)=v_1=11\ldots
100\ldots 0$  for each $X\in\C_1$, be a sub-code for which
$R(\C_1)$ attains the  upper  bound of
Theorem~\ref{thm:upper_rank}, i.e., $|\C_1| = | R(\C_1)|=
q^{(k-\delta+1)(n-k)}$. Then there is no codeword $Y$ in $\C$ such
that $d_H(v(Y),v_1)<d$.
\end{lemma}

\begin{proof}
Let $\C$ be a given $(n,M,d=2\delta,k)_q$ constant dimension code.
Since the minimum distance of the code is $d$, the intersection of any
two subspaces in $\C$ is at most of dimension $k-\frac{d}{2}=k-\delta$.
Therefore, a subspace of dimension $k-\delta+1$
can be contained in at most  one codeword of $\C$.

We define the following set of subspaces:
\[A=\{X\in \cG_q(n,k-\delta+1): \;supp(v(X))\subseteq supp(v_1)\},
\]
where $supp(v)$ is as the set of nonzero entries in $v$. Each
codeword of the sub-code $\C_1$ contains $\begin
{footnotesize}\sbinomq {k}{k-\delta+1}\end{footnotesize}$
subspaces of dimension $k-\delta+1$, and all subspaces of
dimension $k-\delta+1$ which are contained in codewords of $\C_1$
are in $A$. Since $|\C_1| = q^{(k-\delta+1)(n-k)}$, it follows
that $\C_1$ contains $q^{(k-\delta+1)(n-k)}\cdot \begin
{footnotesize}\sbinomq {k}{k-\delta+1}\end{footnotesize}$
subspaces of~$A$.

Now we calculate the size of $A$. First we observe that
\[ A= \{X\in \cG_q(n,k-\delta+1): v(X)=ab,\; |a|=k,\; |b|=n-k,\; w(a)=k-\delta+1,\; w(b)=0\},
\]
where $|v|$ and $w(v)$ are the length and the weight of a vector
$v$, respectively. Thus $EF(v(X))$ of each $v(X)=ab$, such that
$X\in A$, has the form
\begin{equation}EF(v(X))=\left[EF(a)
\begin{footnotesize}
\begin{array}{cccc}
\bullet & \bullet &\ldots & \bullet  \\
\bullet & \bullet & \ldots &\bullet  \\
\bullet & \bullet &\ldots  &\bullet  \\
\end{array}
\end{footnotesize}\\\right].\label{EF_of_A}
\end{equation}
The number of dots in (\ref{EF_of_A}) is $(k-\delta+1)(n-k)$, and
the size of the following set $$\{EF(a):|a|=k,\;
w(a)=k-\delta+1\}$$ is $\begin{footnotesize}\sbinomq
{k}{k-\delta+1}\end{footnotesize}$. Therefore, $|A|= \begin
{footnotesize}\sbinomq {k}{k-\delta+1}\end{footnotesize}\cdot
q^{(k-\delta+1)(n-k)}$. Hence, each subspace of $A$ is contained
in some codeword from $\C_1$. A subspace $Y\in \Gr$ with
$d_H(v(Y),v_1)=2\delta-2i$, $1\leq i\leq \delta-1$, contains some
subspaces of  $A$, and therefore, $Y\notin\C$.
\end{proof}

\begin{lemma}
\label{lm:second_id} Let $\C$ be an $(n,M,d=2\delta,k)_q$ constant
dimension code, where $\delta-1\leq k-\delta$. Let  $\C_2$ be a
sub-code of $\C$ which corresponds to the identifying vector $v_2
= abfg$, where $a=\underset{k-\delta}{\underbrace{11\ldots 1}}$,
$b=\underset{\delta}{\underbrace{00\ldots 0}}$,
$f=\underset{\delta}{\underbrace{11\ldots 1}}$, and
$g=\underset{n-k-\delta}{\underbrace{00\ldots 0}}$. Assume further
that $R(\C_2)$ attains the  upper bound of
Theorem~\ref{thm:upper_rank}, i.e.,
$|\C_2|=|R(\C_2)|=q^{(k-\delta+1)(n-k)-\delta^2}$. Then there is
no codeword $Y\in \C$  with $v(Y)=a'b'fg'$, $|a'b'|=k$,
$|g'|=n-k-\delta$, such that $d_H(v(Y),v_2)<d$.
\end{lemma}

\begin{proof}
Similarly to the proof of Lemma~\ref{lm:first_id},
we define the following set of subspaces:
\[B=\{X\in \cG_q(n,k-\delta+1):
v(X)=a''bfg \; \textrm{with}\; |a''|=k-\delta, \; w(a'')=k-2\delta+1 \}.
\]
As in the previous proof, we can see that $\C_2$ contains
$q^{(k-\delta+1)(n-k)-\delta^2}\cdot \begin {footnotesize}\sbinomq
{k-\delta}{k-2\delta+1}\end{footnotesize}$ subspaces of $B$. In
addition, $|B|= \begin {footnotesize}\sbinomq
{k-\delta}{k-2\delta+1}\end{footnotesize}\cdot q^{(k-2
\delta+1)\delta+(k-\delta+1)(n-k-\delta)}=\begin
{footnotesize}\sbinomq {k-\delta}{k-2\delta+1}\end{footnotesize}
\cdot q^{(k-\delta+1)(n-k)-\delta^2}$. Thus each subspace in $B$
is contained in some codeword from $\C_2$. A subspace $Y\in \Gr$,
such that $v(Y)=a'b'fg'$ ($|a'b'|=k$, $|g'|=n-k-\delta$), with
$d_H(v(Y),v_2)=2\delta-2i$, $1\leq i\leq \delta-1$, contains some
subspaces of $B$, and therefore, $Y\notin \C$.
\end{proof}

\section{Search for Constant Dimension Lexicodes}
\label{sec:search}

In this section we describe our search method for constant
dimension lexicodes, and present some  resulting codes which are
the largest currently known constant dimension codes for their
parameters.

To search for large constant dimension code we use the multilevel
structure of such codes, described in the previous section. First,
we order the set of all binary words of length $n$ and weight $k$
by an appropriate order. The words in this order are the
candidates to be the identifying vectors of the final code. In
each step of the construction we have the current code $\C$ and
the set of subspaces not examined yet. For each candidate for an
identifying vector $v$ taken by the given order, we search for a
sub-code in the following way: for each subspace $X$ (according to
the lexicographic order of subspaces associated with $v$) with the
given Ferrers diagram we calculate the distance between $X$ and
$\C$, and add $X$ to $\C$ if this distance is at least $d$. By
Theorem~\ref{thm:distance} and Corollary~\ref{cor:distance} it
follows that in this process, for some subspaces it is enough only
to calculate the Hamming distance between the identifying vectors
in order to determine a lower bound on the subspace distance. In
other words, when we examine a new subspace to be inserted into
the lexicode, we first calculate the Hamming distance between its
identifying vector and the identifying vector of a codeword, and
only if the distance is smaller than $d$, we calculate the rank of
the corresponding matrix, (see (\ref{subspace_distance})).
Moreover, by the multilevel structure of a code, we need only to
examine the Hamming distance between the identifying vectors of
representatives of sub-codes, say the first codewords in  each
sub-code. This approach will speed up the process of the code
generation.

A construction of constant dimension lexicodes based on the
Ferrers tableaux form ordering of the Grassmannian was give in
~\cite{SiEt09}. Note that in this construction we order the
identifying vectors by the sizes of corresponding Ferrers
diagrams. The motivation is that usually a larger diagram
contributes more codewords than a smaller one.

\begin{example}
\label{exm:8-4-4}
Table \ref{tab:lexicode8-4-4} shows  the identifying vectors and the sizes
of corresponding sub-codes in the $(8,4605,4,4)_2$ lexicode, $\C^{lex}$ (see
\cite{SiEt09}),
and the $(8,4573,4,4)_2$  code, $\C^{ML}$, obtained by the ML
construction~\cite{EtSi09}.

\begin{table}[h]
\centering \caption{ $\C^{lex}$  vs. $\C^{ML}$ in $\mathcal G_2(8,4)$ with
$d_S=4$}\label{tab:lexicode8-4-4}
\begin{tabular}{|c|c|c|c|}
\hline \multicolumn{1}{|c|}{$i$} &
 \multicolumn{1}{|c|}{id.vector $v_i$} &
\multicolumn{1}{c|}{size of $\C^{lex}_i$} &
\multicolumn{1}{c|}{size of $\C^{ML}_i$} \tabularnewline
\hline\hline 1&11110000 & 4096  &4096 \tabularnewline\hline 2&
11001100 &256  & 256 \tabularnewline\hline 3& 10101010 & 64 & 64
\tabularnewline\hline 4& 10011010 & 16 & -- \tabularnewline\hline
5& 10100110 & 16 & -- \tabularnewline\hline 6& 00111100 & 16 & 16
\tabularnewline\hline 7& 01011010 & 16 & 16 \tabularnewline\hline
8& 01100110 & 16 & 16 \tabularnewline\hline 9& 10010110 & 16 & 16
\tabularnewline\hline 10&01101001 & 32 & 32 \tabularnewline\hline
11&10011001 & 16 & 16 \tabularnewline\hline 12&10100101 & 16 & 16
\tabularnewline\hline 13&11000011 & 16 & 16 \tabularnewline\hline
14&01010101 & 8 & 8 \tabularnewline\hline 15&00110011 & 4 & 4
\tabularnewline\hline 16&00001111 & 1 & 1 \tabularnewline\hline
\end{tabular}
\end{table}
We can see that these two codes have the same identifying vectors,
except for two vectors 10011010 and 10100110  in the lexicode
$\C^{lex}$ which form the difference in the size of these two
codes. In  addition, there are several  sub-codes of $\C^{lex}$
for which the corresponding unlifted codes  are nonlinear:
$\C^{lex}_4$, $\C^{lex}_5$, $\C^{lex}_7$, $\C^{lex}_8$,
$\C^{lex}_{11}$, and $\C^{lex}_{12}$. However, all these  unlifted
codes are  cosets of linear codes.
\end{example}
\ \\
In general, not all unlifted codes of lexicodes based on the
Ferrers tableaux form representation are linear or cosets of some
linear codes. However, if  we construct a binary constant
dimension lexicode with only one identifying vector, the unlifted
code is always linear. This phenomena can be explained as an
immediate consequence from the main theorem in~\cite{Zan97}.
However, it does not explain why some of unlifted codes in
Example~\ref{exm:8-4-4} are cosets of linear codes, and why
$\C_9^{lex}$ is linear ($d_H(v_5,v_9)<4$)?


Based on Theorem \ref{thm:distance}, Lemmas
\ref{lm:first_id}, and \ref{lm:second_id}, we suggest  an improved search of
a constant dimension  $(n,M,d,k)_q$ code,  which
will be called a \textit{lexicode with a seed}.

In the first step we construct a  maximal sub-code $\C_1$  which
corresponds to the identifying vector
$\underset{k}{\underbrace{11\ldots 1}}
\underset{n-k}{\underbrace{00\ldots 0}}$. This sub-code
corresponds to the largest Ferrers diagram. In this step we can
take any known $[k\times(n-k), (n-k)(k-\frac{d}{2}+1),
\frac{d}{2}]$ MRD code (e.g.~\cite{Gab85}) and consider its
codewords as the unlifted codewords (Ferrers tableaux forms) of
$\C_1$.

In the second step we construct a  sub-code $\C_2$  which
corresponds to the  identifying vector $\underset{k-\delta}
{\underbrace{11\ldots 1}}\underset{\delta}{\underbrace{00\ldots
0}} \underset{\delta}{\underbrace{11\ldots 1}}
\underset{n-k-\delta} {\underbrace{00\ldots 0}}$. According to
Lemma~\ref{lm:first_id}, we cannot use identifying vectors with
larger Ferrers diagrams (except for the identifying vector
$\underset{k}{\underbrace{11\ldots 1}}
\underset{n-k}{\underbrace{00\ldots 0}}$ already used). If there
exists an MRD (Ferrers diagram) code with the corresponding
parameters, we can take any known construction of such code (see
in \cite{EtSi09}) and build from it the corresponding sub-code. If
a code which attains the bound of Theorem~\ref{thm:upper_rank} is
not known, we take the largest known Ferrers diagram rank-metric
code with the required parameters.

In the third step we construct the other sub-codes, according to the
lexicographic order based on the Ferrers tableaux form representation.
We first calculate the Hamming distance between the identifying vectors and
examine the subspace distance only of subspaces which are not pruned out by
Lemmas \ref{lm:first_id} and \ref{lm:second_id}.

\begin{example}
\label{exm:10-5-6}

Let $n=10$, $k=5$, $d=6$, and $q=2$.
By the construction of a lexicode with a seed we obtain  a constant
dimension code of size $32890$. (A code of
size $32841$ was obtained by the ML construction~\cite{EtSi09}).
\end{example}

\begin{example}
Let  $n=7$, $k=3$, $d=4$, and $q=3$. By the construction a lexicode
with a seed we obtain  a constant
dimension code of size $6691$. (A code of
size $6685$ was obtained by the ML construction~\cite{EtSi09}).
\end{example}

We introduce now  a variant of the construction of a lexicode with
a seed. As a seed we take a constant dimension code obtained by
the ML construction \cite{EtSi09} and try to add some more
codewords using the lexicode construction. Similarly, we can take
as a seed any subset of codewords obtained by any given
construction and to continue by applying the lexicode with a seed
construction.

\begin{example}
Let $n=8$, $k=d=4$, and $q=2$. We take the $(8,4573,4,4)_2$ code
obtained by the
ML construction (see Table~\ref{tab:lexicode8-4-4})
and then continue with the lexicode
construction. The size of the resulting code is $4589$
(compared to $\C^{lex}$ of size $4605$ in Table~\ref{tab:lexicode8-4-4}),
where there are two additional sub-codes of size $8$ which
correspond to identifying vectors $10011010$  and $10100110$.
\end{example}

\begin{example}
\label{ex:4-4-9} Let $n=9$, $k=d=4$, and $q=2$. Let $\C$ be a
$(9,2^{15}+2^{11}+2^7,4,4)_2$ code obtained as follows. We take
three codes of sizes $2^{15}$, $2^{11}$, and $2^7$, corresponding
to identifying vectors $111100000$, $110011000$, and $110000110$,
respectively, and then continue by applying the lexicode with a
seed construction. For the identifying vector $111100000$ we can
take as the unlifted code, any code which attains the bound of
Theorem~\ref{thm:upper_rank}. To generate the codes for the last
two identifying vectors with the corresponding unlifted codes
(which attains the bound of Theorem~\ref{thm:upper_rank}), we
permute the order of entries in the Ferrers diagrams and apply the
lexicode construction.
 The Ferrers diagrams which correspond to the
identifying vector $110011000$  and $110000110$ are
$$
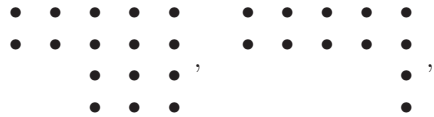
\begin{array}{ccccc}
 \bullet & \bullet & \bullet & \bullet & \bullet\\
 \bullet & \bullet & \bullet & \bullet & \bullet \\
 & & \bullet & \bullet  & \bullet  \\
 &  & \bullet & \bullet & \bullet  \\
\end{array},\;\;\;
\begin{array}{ccccc}
 \bullet & \bullet & \bullet & \bullet & \bullet\\
 \bullet & \bullet & \bullet & \bullet & \bullet \\
 & & & & \bullet  \\
 &  & & & \bullet  \\
\end{array},
$$
respectively. The coordinates' order of their entries (defined in
the Introduction) is:
$$
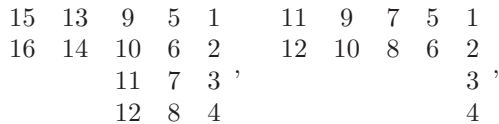
\begin{array}{ccccc}
15 & 13 & 9 & 5& 1\\
16 & 14 & 10 & 6& 2\\
 & & 11 & 7 & 3  \\
 &  & 12 & 8 & 4 \\
\end{array},\;\;\;
\begin{array}{ccccc}
11 & 9 & 7 & 5& 1\\
12 & 10 & 8 & 6& 2\\
 & &  &  & 3  \\
 &  &  &  & 4 \\
\end{array},
$$
respectively. The order of the coordinates that we use to form an MRD code (lexicode) is
$$\begin{array}{ccccc}
11 & 7  & 5 & 3 & 1\\
15 & 12 & 8 & 2 & 4\\
 & &     13 & 9 & 6 \\
 & &     16 &14 & 10 \\
\end{array},\;\;
\begin{array}{ccccc}
9 & 7  & 5 & 3 & 1\\
11 & 10 & 8 & 2 & 4\\
 & &      &  & 6 \\
 & &      & & 12 \\
\end{array}.
$$
As a result, we obtain a code of  size $37649$ which is the
largest known constant dimension code with these parameters.
\end{example}

\begin{remark}
One of the most interesting questions, at least from a
mathematical point of view, is the existence of a $(7,381,4,3)_2$
code $\C$~\cite{EV09}. If such code exists one can verify that it
contains $128$ codewords with the identifying vector $1110000$
which is half the size of the corresponding MRD code. It suggests
that the unlifted Ferrers diagram rank-metric code of the largest
Ferrers diagram is not necessarily  an MRD code, in the largest
constant dimension code with given parameters $n$, $k$, and $d$.
\end{remark}

\section{Conclusion and Open Problems}
\label{sec:conclusion}

We have described a search method for constant dimension codes
based on their multilevel structure. Some of the codes obtained by
this search are the largest known constant dimension codes with
their parameters. We described several ideas to make this search
more efficient. In this context a new formula for computation of
the subspace distance between two subspaces of $\F_q^n$ is given.
It is reasonable to believe that the same ideas will enable to
improve the sizes of the codes with parameters not considered in
our examples. We hope that a general mathematical technique to
generate related codes with larger size can be developed based on
our discussion. Our discussion raises several more questions for
future research:

\begin{enumerate}
\item Is the upper bound of Theorem~\ref{thm:upper_rank} on the
size of Ferrers diagram rank-metric code  is attainable for all
parameters?

\item What is the best choice of identifying vectors for constant
dimension lexicode in general, and for the the ML construction in
particular?

\item Can every MRD Ferrers diagram code  be generated as a
lexicode by using  a proper permutation on the coordinates (see
Example~\ref{ex:4-4-9})?

\item Is there an optimal combination of linear Ferrers diagram
rank-metric codes and cosets of linear Ferrers diagram rank-metric
codes to form a large constant dimension code?

\item For which $n$ and $k$ there exists an order of all
identifying vectors such that all the unlifted codes (of the
lexicode) will be either linear or cosets of linear codes (see
Example~\ref{exm:8-4-4}).

\end{enumerate}

\medskip
\medskip

\end{document}